\documentclass[a4paper,UKenglish,cleveref,autoref]{lipics-v2021}

\title{Approximate Maximum Halfspace Discrepancy}

\author{Michael Matheny}{Amazon}{mmath@cs.utah.edu}{}{}

\author{Jeff M. Phillips}{University of Utah}{jeffp@cs.utah.edu}{}{Thanks to supported by NSF CCF-1350888, IIS-1251019, ACI-1443046, CNS-1514520, and CNS-1564287}

\authorrunning{M. Matheny and J.\,M. Phillips}

\Copyright{TBA}



\nolinenumbers 

\hideLIPIcs  

\EventEditors{John Q. Open and Joan R. Access}
\EventNoEds{2}
\EventLongTitle{42nd Conference on Very Important Topics (CVIT 2016)}
\EventShortTitle{CVIT 2016}
\EventAcronym{CVIT}
\EventYear{2016}
\EventDate{December 24--27, 2016}
\EventLocation{Little Whinging, United Kingdom}
\EventLogo{}
\SeriesVolume{42}
\ArticleNo{23}

\usepackage{pgf,tikz,tikzscale}
\usetikzlibrary{arrows}

\usepackage{amssymb}
\usepackage{epsfig}
\usepackage{xspace}
\usepackage{color}
\usepackage{wrapfig}
\usepackage{microtype}
\usepackage{mathtools}
\usepackage{amsmath}

\usepackage{algorithm}
\usepackage[noend]{algorithmic}

\usepackage{graphicx}
\usepackage{caption}
\usepackage{float}
\usepackage{wrapfig}

\graphicspath{{IpeFigures/}{Figures/}}

\newcommand{\E}{\ensuremath{\mathsf{E}}}

\newcommand{\eps}{\varepsilon}
\newcommand{\poly}{\ensuremath{\mathsf{poly}}}

\renewcommand{\c}[1]{\ensuremath{\mathcal{#1}}}

\newcommand{\R}{\mathbb{R}}

\newcommand{\omt}[1]{}
\newcommand{\etal}{\emph{et al.}\xspace}

\renewcommand{\c}[1]{\ensuremath{\mathcal{#1}}}

\newcommand{\eMH}{\textsc{$\eps$-Max-Halfspace}\xspace}
\newcommand{\MH}{\textsc{Max-Halfspace}\xspace}
\newcommand{\PC}{\textsc{Point-Covering}\xspace}
\newcommand{\LC}{\textsc{Line-Covering}\xspace}
\newcommand{\MWKC}{\textsc{Max-Weight-$K$-Clique}\xspace}
\newcommand{\MWTC}{\textsc{Max-Weight-$3$-Clique}\xspace}
\newcommand{\MWP}{\textsc{Max-Weight-Point}\xspace}
\newcommand{\TSUM}{\textsc{3Sum}\xspace}
\newcommand{\APSP}{\textsc{APSP}\xspace}
\newcommand{\NT}{\textsc{Negative-Triangle}\xspace}

\begin{document}

\maketitle

\begin{abstract}
Consider the geometric range space $(X, \c{H}_d)$ where $X \subset \mathbb{R}^d$ and $\c{H}_d$ is the set of ranges defined by $d$-dimensional halfspaces. In this setting we consider that $X$ is the disjoint union of a red and blue set. For each halfspace $h \in \c{H}_d$ define a function $\Phi(h)$ that measures the ``difference'' between the fraction of red and fraction of blue points which fall in the range $h$.   
In this context the maximum discrepancy problem is to find the $h^* = \arg \max_{h \in (X, \c{H}_d)} \Phi(h)$. 
We aim to instead find an $\hat{h}$ such that $\Phi(h^*) - \Phi(\hat{h}) \le \eps$. 
This is the central problem in linear classification for machine learning, in spatial scan statistics for spatial anomaly detection, and shows up in many other areas.  
We provide a solution for this problem in $O(|X| +  (1/\eps^d) \log^4 (1/\eps))$ time, which improves polynomially over the previous best solutions.   For $d=2$ we show that this is nearly tight through conditional lower bounds.  For different classes of $\Phi$ we can either provide a $\Omega(|X|^{3/2 - o(1)})$ time lower bound for the exact solution with a reduction to \textsc{APSP}, or an $\Omega(|X| + 1/\eps^{2-o(1)})$ lower bound for the approximate solution with a reduction to \textsc{3Sum}.  

A key technical result is a $\eps$-approximate halfspace range counting data structure of size $O(1/\eps^d)$ with $O(\log (1/\eps))$ query time, which we can build in $O(|X| + (1/\eps^d) \log^4 (1/\eps))$ time.  
\end{abstract}


\section{Introduction}

Let $X$ be a set of $m$ points in $\R^d$ for constant $d$ where $X$ can either be the union of a red, $R$, and blue set, $B$, of points $X = R \cup B$ (possibly not disjoint) or a set of weighted points where each point has weight $w(x)$ for $x \in X$.  Now let $(X, \c{H}_d)$ be the associated range space of all subsets of $X$ defined by intersection with a halfspace; the halfspaces are not restricted to go through the origin in this paper.

We are interested in finding the halfspace $h^*$ and value $\Phi^*$ that maximizes a function $\Phi_X(h) : \c{H}_d \rightarrow \R$ for some class of functions $\Phi$.  We characterize them by reframing it as a function of $\mu_R$ and $\mu_B$ so $\Phi_X(h) = \phi(\mu_R(h), \mu_B(h))$,
where $\mu_R(h) = |R \cap h|/|R|$ and $\mu_B(h) = |B \cap h|/|B|$ are the fraction of red or blue points, respectively, in the range $h$. 
In particular, we only consider functions $\Phi_X(h)$ which can be calculated in $O(1)$ time  from $\mu_R(h)$ and $\mu_B(h)$ as $\phi(\mu_R(h), \mu_B(h))$ (e.g., $\phi(\mu_R, \mu_B) = |\mu_R - \mu_B|$).  
Given such a fixed $\phi$, or one from a class, we state the two main problems: exact and $\eps$-additive error.  

\begin{itemize}
\item  \textbf{Problem} \MH: \emph{From a given set $X = R \cup B \subset \R^d$ points where $|X|= m$ and a Lipshitz constant function $\Phi_X(h) : 2^X \rightarrow \R$, find $h^*=\arg \max_{h \in \c{H}_d} \Phi_X(h)$.} 


\item  \textbf{Problem} \eMH: \emph{From a given set $X = R \cup B \subset \R^d$ points where $|X|= m$ and a Lipshitz constant function $\Phi_X(h) : 2^X \rightarrow \R$ where $h^*=\arg \max_{h \in \c{H}_d} \Phi_X(h)$, find $\hat{h} \in \c{H}_d$ such that 
$\Phi_X(h^*) - \Phi_X(\hat{h}) \le \eps $.} 
\end{itemize}

\begin{figure}[t]
\vspace{-2mm}
\includegraphics[width=0.325\linewidth]{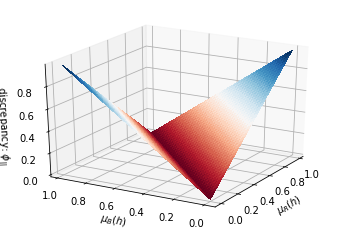}
\phantom{}
\includegraphics[width=0.325\linewidth]{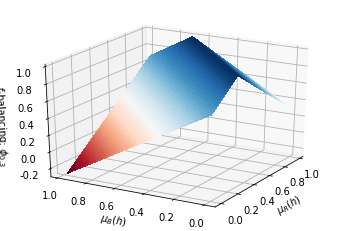}
\phantom{}
\includegraphics[width=0.325\linewidth]{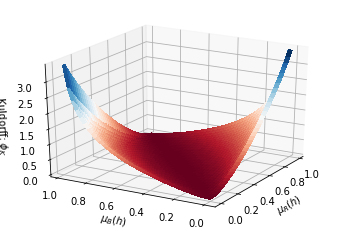}
\vspace{-2mm}
\caption{Plots of common $\phi$ functions.  From left to right: $\phi_{||}$, $\phi_f$ (for $f=0.3$), and $\phi_K$.}
\label{fig:phi}
\end{figure}

We typically write $\Phi_X$ as $\Phi$ when the set $X$ is clear.  Our algorithms will explicitly compute $\mu_R(h)$ (and $\mu_B(h)$ separately) for each $h$ evaluated, and so work with all functions $\phi$.   Notable useful examples of $\phi$, shown in Figure \ref{fig:phi}, include:
\begin{itemize}
\item \textbf{discrepancy:} $\phi_{||}(h) = |\mu_R(h) - \mu_B(h)|$.  \\
This measures the maximum disagreement between the proportions of the two sets $B$ and $R$: the central task of building a linear classifier (on training data) in machine learning is often formulated to maximize precisely this function~\cite{LM96,NS13,diakonikolas2019distribution,diakonikolas2020learning}.  Given a coloring of $X$ into $R$ and $B$, this also measures the discrepancy of that coloring~\cite{Mat99,Cha00}.  
\item \textbf{$f$-balancing:} $\phi_f(h) = 1 - | (\mu_R(h) - \mu_B(h)) - f|$, for a fraction $f \in (0,1)$.  \\
This scores how well a halfspaces strikes a balance of $f$ between the two sets.  By minimizing this function (or maximizing $1-\phi_f(h)$), the goal is to find a range $h$ that exhibits an imbalance between the sets of $f$.  Maximizing this function (say with $f=1/2$) is the problem of finding good ham sandwich cuts~\cite{LMS94}.  
\item \textbf{Kulldorff:} $\phi_K(h) = \mu_R(h) \log \frac{\mu_R(h)}{\mu_B(h)}  + (1 - \mu_R(h)) \log \frac{1- \mu_R(h)}{1 - \mu_B(h)}$.  \\
This is the Kulldorff discrepancy function~\cite{Kul97} which arises in spatial scan statistics~\cite{MP18a,SSSS,Kul97,Kul7.0,HKG07,NM04,APV06,AMPVZ06}, for detecting spatial anomalies.  This function specifically is derived as the log-likelihood ratio test under a Poisson model, but other similar convex, non-linear functions arise naturally from other models~\cite{Kul7.0,AMPVZ06}.  In this setting the most common range shape model is a disk, and the best algorithms~\cite{DE93,DEM96,MP18a} operate by lifting to one-dimension higher where the ranges correspond to halfspaces.  
\end{itemize}

All of these $\phi$  functions are Lipschitz continuous over $\mu_B$ and $\mu_R$, and thus this extends to a combinatorial notion of Lipschitz over the associated $\Phi$ with respect to the combinatorial range $h$: that is, 
$|\Phi_X(h) - \Phi_X(h')| \le c (||h \cap R| - |h' \cap R||/|R| + ||h \cap B| - |h' \cap B||/|B|)$ 
for constant $c$.  For $\phi_{||}$ and $\phi_f$, $c=1$, and for $\phi_K$ it is bounded in a reasonable range of $\mu_B,\mu_R$~\cite{AMPVZ06,MP18b}.  
This means that approximating each of $\mu_B(h)$ and $\mu_R(h)$ up to additive error translates into at most additive error in $\Phi_X$.  

We also consider the hardness of these problems, and for this we will restrict the classes of functions $\phi$ considered in these problems.  
We will show the largest lower bound on \emph{concave} functions $\phi$ (like $\phi_f$), and this construction will apply to a Lipschitz functions (again like $\phi_f$).
However, the class of convex functions (like $\phi_{||}$ and $\phi_K$) are more prevalent, and we present a smaller lower bound, but which applies to convex Lipschitz functions (including $\phi_{||}$ and $\phi_K$).

We also note that for all of the above functions and linked challenges, the $\eps$-approximate versions are just as relevant and common as the exact ones.  
For instance in machine learning, typically $\eps$-additive error is the baseline, on assumptions that the input $X$ is drawn from a fixed but unknown distribution~\cite{VC71,LLS01}.  
Spatial scan statistics are applicable to data sets containing thousands to 100s of millions of spatial data points, such as census data, disease incidents, geo-located social media posts.  The exact algorithms are polynomial in $m$ and therefore can have massive runtimes, which become infeasible in the full data sets.  Usually the exact algorithms are not even attempted and the set of ranges considered are defined using some heuristic such as disks with centerpoints of a grid or centered at data points.  
On the other hand the $\eps$-additive error versions are scalable with runtimes of $O(m + \poly(\frac{1}{\eps}))$, depending only on the accuracy of the solution not the scale of the data.

\subparagraph*{Our Results.}
We connect this problem to the approximate range counting problem for $d \ge 2$ in the additive error model by designing a data structure of size $O(\frac{1}{\eps^d})$ that can be constructed in time $O(m + \frac{1}{\eps^d} \log^4 \frac{1}{\eps})$ with constant probability and supports range counting queries in time $O(\log \frac{1}{\eps})$. 
This structure implies a halfspace scanning algorithm that runs in time $O(m + \frac{1}{\eps^d} \log^4 \frac{1}{\eps})$. The data structure is closely related to cuttings, but we do not need the set of crossing lines, but only an approximation of their total count.  

In the other direction, we show instances where $\phi$ is linear that are as hard as $\Omega(m^{3/2 - o(1)})$ on the exact problem, reduced from \APSP, and instances where $\phi$ is concave that are as hard as $\Omega(m + \frac{1}{\eps^{2 - o(1)}})$, reduced from \TSUM. 
This implies (conditionally) that this class of scanning algorithms requires $\Omega(m + \frac{1}{\eps^{2 - o(1)}})$ time, and any further algorithmic improvements (beyond $\mathrm{polylog}(1/\eps))$ factors) would either require specific and new assumptions on $\phi$ or improvement on a classic hard problem.

\subparagraph{Relation to Prior Work.}
The best prior algorithms for \eMH required $O(m + (1/\eps)^{d + 1/3} \log^{2/3}(1/\eps))$ time~\cite{MP18b} and the best exact algorithm requires $O(m^d)$ time~\cite{DE93,DEM96}.  
Conditioned on \TSUM, without allowing restrictions to $\phi$ beyond linearity, our lower bound shows the prior exact algorithms are shown tight for $d=2$ by setting $m=1/\eps$.  

For the related problem of the range space defined by axis-aligned rectangles~\cite{wei2018tight}, the story for instance is more complicated with respect to $\phi$. 
For linear $\phi$ the exact problem can be solved in $O(m^2)$ time~\cite{Barbay2014} which is tight~\cite{Backurs16} assuming no subcubic algorithm for \MWTC (and hence \APSP).  
The $\eps$-approximate version can be solved in $O(m + \frac{1}{\eps^{2d-2}} + \frac{1}{\eps^2} \log \log \frac{1}{\eps})$ time with constant probability~\cite{MP18b}, and by the result of Backurs \etal~\cite{Backurs16} this cannot be improved beyond $\Omega(m + 1/\eps^2)$ in $\R^2$ under the same assumptions~\cite{MP18b}.  
However, for general functions $\phi$ the best known runtimes increase to $O(m^4)$ and $O(m + 1/\eps^4)$ for the exact and approximate versions, respectively~\cite{MP18b}.  Although for a big class of convex functions (like $\phi_K)$ can be reduced back to $O(m + 1/\eps^{2.5})$ in the approximate case~\cite{MP18b}.  
Thus, this paper shows the situation appears significantly simpler for the halfspaces case, and we provide a new connection to a (usually) separate~\cite{VWW18,Backurs16,VW18} class of conditional hardness problems through \TSUM.

Our approximate range counting structure is also new and may be of independent interest,  allowing very fast halfspace queries (in $O(\log 1/\eps)$ time) in moderate dimensions $d$ where $1/\eps^d$ may not be too large.  For example, $3d$ disk queries maps to the $d=4$ setting, and with $\eps = 0.01$ (1\% error) then $1/\eps^d = 100$ million -- which can fit in the memory of most modern systems, and allow for very fast queries.  
Similar structures are possible for the exact range counting paradigm~\cite{BS95,Cha00}, and these could be adapted to the approximate setting after constructing an appropriate $\eps$-sample of $X$~\cite{VC71,LLS01,MP18a}.  But these would instead use higher $\tilde{O}(1/\eps^{d+1})$ preprocessing time (where $\tilde{O}(z)$ hides $\mathsf{polylog}(z)$ terms).  
Most effort in approximate range counting has come in the low-space regime.  For instance with \emph{relative} $(1+\eps)$ error, one can build a data structure of expected size $O(\mathrm{poly}(1/\eps) \cdot m)$ that answers approximate range counting queries in $O(\mathrm{poly}(1/\eps) \cdot \log m)$ time in $\R^3$~\cite{AC09}.  
However for $d > 3$ with linear space, the query time becomes polynomial in $m$ at $\tilde O(m^{1-1/\lfloor d/2 \rfloor})$~\cite{AH08,RS-survey17,Sal17}.

A related line of work on robust halfspace learning under specific noise models in high dimensions has witnessed significant progress in the last few years~\cite{diakonikolas2019distribution,diakonikolas2020learning,diakonikolas2021agnostic,chen2020classification}.  
These models and algorithms also start with a sufficiently large iid sample (of size roughly $1/\eps^2$), and ultimately achieve a result with additive $\eps$ error from the opt.  However, they assume a specific noise model, but using this achieve runtimes polynomial in $d$, whereas our results grow exponentially with $d$ but do not require any assumptions on the noise model.   
A recent focus is on Massert noise, where given a perfect classifier, each point has its sign flipped independently with an assigned probability at most $\eta < 1/2$.  Only recently~\cite{chen2020classification}, in this Massert noise model it was achieved a proper learned function (a halfspace) in time polynomial in $1/\eps$ and in $d$.  Recent previous work was not proper (the learned function may not be a halfspace)~\cite{diakonikolas2019distribution,diakonikolas2020learning}, or takes time exponential in $1/\eps$~\cite{diakonikolas2021agnostic}.  Indeed, under Gaussian distributed data, the proper halfspace learning requires $d^{\Omega(1/\eps)}$ time~\cite{diakonikolas2020near}, under the statistical query model~\cite{reyzin2020statistical}.

\section{Background and Notation}
\label{sec:review}

\subparagraph*{Useful Sampling Properties.}
An \emph{$\eps$-sample}~\cite{HPbook,VC71} $S \subset X$ of a range space $(X,\c{R})$ preserves the density for all ranges as 
$
\max_{A \in \c{R}} |\frac{|X \cap A|}{|X|} - \frac{|S \cap A|}{|S|}| \leq \eps.  
$
An \emph{$\eps$-net} $N \subset X$ of a range space $(X,\c{R})$ hits large ranges, specifically for all ranges $A \in \c{R}$ such that $|X \cap A| \geq \eps |X|$ we guarantee that $N \cap A \neq \emptyset$.  
For halfspaces, 
a random sample $S \subset X$ of size $O(\frac{1}{\eps^2} (d + \log \frac{1}{\delta})$ is an $\eps$-sample with probability at least $1-\delta$~\cite{VC71,LLS01}, and a
 random sample $N \subset X$ of size $O(\frac{d}{\eps} \log \frac{1}{\eps \delta})$ is an $\eps$-net~\cite{HW87} with probability at least $1-\delta$.  
\subparagraph*{Enumeration.}
Given a set of points $X \subset \mathbb{R}^d$, Sauer's Lemma~\cite{Sau72} shows that there are at most $O(|X|^d)$ combinatorially distinct ranges, where each range defined by a halfspace contains the same subset of points.  
We can always take a halfspace and rotate it until it intersects at most $d$ boundary points without changing the set of points contained inside. This observation immediately implies a simple algorithm for scanning the point set; we can just enumerate all 
subsets of at most $d$ points, compute a halfspace that goes through these points, and count the red and blue points lying underneath to evaluate $\Phi$.

Our work builds upon this simple algorithm by dividing it into two steps and optimizing both. We first define a set of prospective ranges $\hat{\c{H}}_d$ to scan and then secondly compute the function $\Phi$ on each region. 
Matheny \etal~\cite{MP18b} showed that a simple random sample $X_0 \subset X$ of size $O(\frac{1}{\eps})$ (for constant $d$) induces a small range space $(X_0, \c{H}_d)$.  Each range (a subset of $X_0$) maps to a canonical geometric halfspace $h_0$, and this geometric halfspace in turn induces a range $h_0 \cap X$, an element of $(X,\c{H}_d)$.  We refer to this subset of $(X,\c{H}_d)$ as $(X,\hat{\c{H}}_d)$, it is of size $O(1/\eps^d)$. Now each range $h \cap X$ in $(X,\c{H}_d)$ there is a range $\hat h \cap X \in (X,\hat{\c{H}}_d)$ such that the symmetric difference between them is at most $\eps|X|$.  
So if every range in $(X,\hat{\c{H}}_d)$ needs to be explicitly checked, and $f(1/\eps)$ is the time to compute $\Phi$, this would imply a $(1/\eps^d)f(1/\eps)$ lower bound.  Note that $f(1/\eps)$ may take super-constant time because we may need to construct the (approximate) $\mu_R(h)$ and $\mu_B(h)$ values.

\subparagraph*{Cuttings.}
Cuttings are a useful tool to act as a divide and conquer step for geometric algorithm design. Given $\mathbb{R}^d$, a set of halfspaces $\c{H}_d$ of size $m$, and some parameter $r$ a $\frac{1}{r}$-\textit{cutting} is a partition of $\mathbb{R}^d$ into a disjoint set of constant complexity cells where each cell is crossed by at most
$m / r$ halfspaces. The set of halfspaces crossing a cell in the cutting is referred to as the conflict list of the cell. It is well established that the number of partitions is of size $O(r^d)$ and the partitioning can be constructed in $O(r^{d - 1} m)$ time if the conflict lists are needed~\cite{Mat91}. If the crossing information is not needed then faster algorithms can be used. For instance, the arrangement of a $\frac{1}{r}$-net over $\c{H}_d$ defines a partitioning of $\mathbb{R}^d$ into disjoint cells where each cell is crossed
by at most $m / r$ halfspaces, and since a simple random sample can be used to generate the net, a cutting of size $O(r^d\log^d r)$ can be computed in $O(m + r^d\log^d r)$ time with constant probability.

When a cutting is restricted to a single cell there are better bounds on the size of the partitioning. We will need this better bound for our proof and we restate it here. 
\begin{theorem}[\cite{Cha93, BS95}]
\label{thm:cutting}
Denote the vertices corresponding to $d$-way intersections of $\c{H}_d$ as $\c{A}(\c{H}_d)$.   A $\frac{1}{r}$-cutting of a cell $\Delta$ containing $|\c{A}(\c{H}_d) \cap \Delta| = \eta$ vertices can be constructed with 
$\displaystyle{O(\eta \left (\frac{r}{m} \right )^d + r^{d - 1})}$  cells.
\end{theorem}

\section{Approximate Halfspace Range Counting and the Upper Bound}
\label{sec:upper}

In this section, instead of operating on $R \cup B = X \in \R^d$, as is common for halfspace range searching, we work on the set $H$ of dual halfspaces in $\R^d$.  In this setting a query halfspace $h \in \c{H}_d$ in the dual representation is $q_h \in \R^d$, and the desired quantity is the number of halfspaces in $H$ below $q_h$.   We apply the construction separately for $R$ and $B$, so the halfspaces $H$ may be weighted, with positive weights.  

We will construct $L$ decompositions of $\R^d$ into disjoint trapezoidal cells: $\mathbf{\Delta}_0, \mathbf{\Delta}_1, \ldots, \mathbf{\Delta}_L$.  
For each level of cells $\mathbf{\Delta}_i$ any cell $\Delta \in \mathbf{\Delta}_i$ has a set of children cells $S_{\Delta}$. The next level of cells $\mathbf{\Delta}_{i + 1}$ is the disjoint union of all the $S_{\Delta}$ children cells of $\Delta \in \mathbf{\Delta}_i$; that is $\mathbf{\Delta}_{i + 1} = \dot{\bigcup}_{\Delta \in \mathbf{\Delta}_i} S_{\Delta}$. Initially $\mathbf{\Delta}_0 = \mathbb{R}^d$ and therefore contains the entire domain and a corresponding sample of this initial cell will be denoted as $\hat{H}$. We will define $\Delta \sqcap H$ to be the set of halfspaces in $H$ that lie completely underneath $\Delta$ and $\Delta \cap H$ is the set of halfspaces in $H$ that cross $\Delta$. 

Importantly, we maintain an estimate of the weight of each cell $m_i(\Delta)$, as we recursively build the decomposition.  It depends on a sufficiently large constant $r$.  
For each cell $\Delta' \in \mathbf{\Delta}_{i + 1}$ where $\Delta' \in S_{\Delta}$ we take a sample $H_{\Delta'} \subset (H_\Delta \cap \Delta')$  of size $\frac{|H_\Delta \cap \Delta'|}{r}$ with replacement. 
The value $\hat m_{i+1}(\Delta')$ estimates the number of halfplanes lying below it, and is defined recursively as $\hat{m}_{i+1}(\Delta') = r^{i+1} |H_{\Delta} \sqcap \Delta'| + \hat{m}_{i}(\Delta)$; with $\hat m_0(\Delta) = 0$. 
If a cell $\Delta$ has a small number of lines crossing it, specifically if $|H_{\Delta}| \le \log \frac{1}{\eps}$ then the recursion terminates.  
Otherwise, we \emph{split} the cell, and create a $1/t_\Delta$-cutting of each $(\Delta, H_\Delta)$, and let $S_\Delta$ be its cells; the value $t_\Delta = \max(\frac{|H_{\Delta}| r^{2i + 1} }{|\hat{H}|}, 1)$ is $1$ if $|H_\Delta|$ is small, otherwise it is at most $r$.  
We recurse on each cell $\Delta'$ in each $S_\Delta$ until each is sufficiently small, which requires $L = O(\log \frac{1}{\eps})$ levels.  
See Algorithm \ref{alg:refine} for details.

\begin{algorithm}
	\caption{SampleCut$(\mathbf{\Delta}_0, \hat H)$}
	\label{alg:refine}
	\begin{algorithmic} 
		\FOR{levels $i = [0, L]$}
			\FOR{each cell in that level $\Delta \in \mathbf{\Delta}_i$ with $|H_{\Delta}| > \log \frac{1}{\eps}$}
					\STATE Set $t_{\Delta} = \max(\frac{|H_{\Delta}| r^{2i + 1} }{|\hat{H}|}, 1)$  
					\STATE Build $S_{\Delta}$, the cells of a $\frac{1}{t_{\Delta}}$-cutting on $(\Delta, H_{\Delta})$.  
					\FOR{each child cell $\Delta' \in S_{\Delta}$}
							\STATE $\hat{m}_{i+1}(\Delta') = r^{i+1} |H_{\Delta} \sqcap \Delta'| + \hat{m}_{i}(\Delta)$
							\STATE Sample $H_{\Delta'} \subset H_{\Delta}$ where $|H_{\Delta'}| = |\Delta' \cap H_{\Delta}| / r$
					\ENDFOR
			\ENDFOR
		\ENDFOR
	\end{algorithmic}
\end{algorithm}

\subparagraph{Complexity analysis.}
As cells are subsampled and split the number of lines and vertices lying inside of a cell drops off quickly with the level.  
\begin{lemma}
	\label{lem:cellsize1}
	A cell $\Delta \in \mathbf{\Delta}_{i}$ with sample $H_{\Delta}$ is of size $|H_{\Delta}| \le |\hat{H}| / r^{2i}$.
\end{lemma}
\begin{proof}
Consider a cell $\Delta^* \in \mathbf{\Delta}_{i - 1}$ where $\Delta \in S_{\Delta^*}$ then $|H_{\Delta}| = \frac{1}{r} |H_{\Delta^*} \cap \Delta|$ by construction,
and since $\Delta$ is a cell in a $\frac{1}{t_{\Delta^*}}$-cutting of $H_{\Delta^*}$ then 
$|H_{\Delta^*} \cap \Delta| \le |H_{\Delta^*}| / t_{\Delta^*} = |\hat H| / r^{2i-1}$; 
hence 
$|H_{\Delta}| = |\Delta \cap H_{\Delta^*}|/r \leq |\hat H| / r^{2i}$.  
\end{proof}

\begin{lemma}
	\label{lem:vertexsize}
	A cell $\Delta \in \mathbf{\Delta}_{i}$ with sample $H_{\Delta}$ has $\E[|\c{A}(H_{\Delta})|] \le \frac{|\c{A}(\hat{H}) \cap \Delta|}{r^{di}}$ expected vertices.
\end{lemma}
\begin{proof}
Consider a vertex lying inside of $\Delta$ induced by the intersection of $d$ halfspaces $h_1, \ldots, h_d \in H_{\Delta^*}$ in $\Delta$ (where $\Delta \subset \Delta^* \in \mathbf{\Delta}_{i-1}$). The probability that this vertex is in $H_{\Delta}$ is 
$
\Pr(h_1 \in H_{\Delta} \wedge \ldots \wedge h_d \in H_{\Delta}) 
\le \frac{1}{r^d}.
$ 
The probability that a vertex survives through $i$ samples is then upper bounded by $\frac{1}{r^{di}}$ and by linearity of expectation $\E[\c{A}(H_{\Delta})] \le \frac{|\c{A}(\hat{H}) \cap \Delta)|}{r^{d i}} $.
\end{proof}

Combining these results we show the expected number of cells increases as $O(r^{di})$.   This leverages Theorem \ref{thm:cutting} using the number of vertex dependent bound for size of the cutting.  

\begin{lemma}
\label{lem:cellsize}
At a level $i$ the expected number of cells is $\E[|\mathbf{\Delta}_i|] = O(r^{di})$.  
\end{lemma}
\begin{proof}
The number of cells at a level $i$ is in expectation $\E[|\mathbf{\Delta}_i|] = \E[\sum_{\Delta \in \mathbf{\Delta}_{i - 1}} O(|\c{A}(H_{\Delta}) \cap \Delta| \frac{t_{\Delta}^d}{|H_{\Delta}|^d} + t_{\Delta}))]$ from Theorem \ref{thm:cutting}.
Since $t_{\Delta} = \max(\frac{|H_{\Delta}| r^{2i + 1} }{|\hat{H}|}, 1)$ we can divide cells 
 in $\mathbf{\Delta}_{i - 1}$ into a set $\mathbf{\Delta}^+_{i - 1}$ where $t_{\Delta} > 1$, and is therefore split, and a set of cells $\mathbf{\Delta}_{i - 1} \setminus \mathbf{\Delta}^+_{i - 1}$ where $t_{\Delta} = 1$, and is therefore not split. The set of non split cells $\mathbf{\Delta}_{i - 1} \setminus \mathbf{\Delta}^+_{i - 1}$ cannot be larger than $\mathbf{\Delta}_{i - 1}$ so in our estimate of $|\mathbf{\Delta}_i|$ these will contribute at most a factor $2$.  

\begin{align*}
\E[|\mathbf{\Delta}_i|] 
=& 
\E\left[\sum_{\Delta \in \mathbf{\Delta}_{i - 1}} O\left(|\c{A}(H_{\Delta}) \cap \Delta| \frac{t_{\Delta}^d}{|H_{\Delta}|^d} + t_{\Delta}^{d-1})\right) \right] 
\\ \le & 
\E \left[\sum_{\Delta \in \mathbf{\Delta}^+_{i - 1}} O\left(|\c{A}(H_{\Delta}) \cap \Delta| \frac{r^{2di-d} }{|\hat{H}|^d} + t_{\Delta}^{d-1}\right)\right]  
\\ =& 
\sum_{\Delta \in \mathbf{\Delta}^+_{i - 1}} O\left(\E[|\c{A}(H_{\Delta}) \cap \Delta| ] \frac{r^{2di-d} }{|\hat{H}|^d} + t_{\Delta}^{d-1}\right) 
\end{align*}
By Lemma \ref{lem:cellsize1} we can bound $t_\Delta \leq \frac{|H_\Delta| r^{2i+1}}{|\hat H|} \leq r$, to replace the second term.  
\[
\le  
\sum_{\Delta \in \mathbf{\Delta}'_{i - 1}} O\left(\E[|\c{A}(H_{\Delta}) \cap \Delta| ] \frac{r^{2di-d} }{|\hat{H}|^d} + r^{d-1}\right) 
\leq
\sum_{\Delta \in \mathbf{\Delta}'_{i - 1}} O\left(\E[|\c{A}(H_{\Delta}) \cap \Delta| ] \frac{r^{2di-d} }{|\hat{H}|^d}\right) + O(r^{d-1}|\mathbf{\Delta}_{i - 1}|) \label{eq:cut_size2}
\]
By Lemma \ref{lem:vertexsize} $\E[|\c{A}(H_{\Delta}) \cap \Delta| ] \le \frac{|\c{A}(\hat{H}) \cap \Delta'|}{r^{di - d}}$, and this yields
\[
\leq
\sum_{\Delta \in \mathbf{\Delta}^+_{i - 1}} O\left(\frac{|\c{A}(\hat{H}) \cap \Delta|}{r^{di -d}} \frac{r^{2di - d} }{|\hat{H}|^d}\right) +  O(r^{d-1}|\mathbf{\Delta}_{i - 1}|) 
= 
O\left (\frac{r^{di}}{|\hat{H}|^d} \right ) \sum_{\Delta \in \mathbf{\Delta}^+_{i - 1}} |\c{A}(\hat{H}) \cap \Delta| + O(r^{d-1}|\mathbf{\Delta}_{i - 1}|).  
\]
Since a vertex in $\c{A}(\hat{H}) \cap \Delta$ can only be in one cell $\sum_{\Delta \in \mathbf{\Delta}^+_{i - 1}} |\c{A}(\hat{H}) \cap \Delta| \le |\hat{H}|^d$, since this quantity upper bounds the number of vertices in $\c{A}(\hat{H})$.  Thus finally
\[
\E[|\mathbf{\Delta}_i|] \le
 \ldots \text{chain of inequalities} \ldots
\le  
O(r^{di}+ r^{d-1}|\mathbf{\Delta}_{i - 1}|) = O(r^{di}).  \qedhere
\]
\end{proof}

\subparagraph{Sampling Error.}
Consider now that we wish to estimate $|\hat H \sqcap \Delta|$, the number of planes crossing under some $\Delta$. We will use that $\Delta$ lies within a nested sequence of cutting cells 
$\Delta = \Delta_i \subset \Delta_{i-1} \subset \ldots \subset \Delta_0$ with $\Delta_i \in \mathbf{\Delta}_i, \Delta_{i-1} \in \mathbf{\Delta}_{i-1}, \ldots, \Delta_0 \in \mathbf{\Delta}_0$, where $i \leq L$, and with corresponding samples
$\hat{H}=H_{\Delta_0}$ and $H_{\Delta_1}, \cdots, H_{\Delta_\ell}$.
%
%
We can also define $\hat m_i(\Delta)$ more generally for a cell $\Delta$ that is the subset of a cell $\Delta_i \in \mathbf{\Delta}_i$ (and its ancestors), but not necessary one of those cells.  It is defined
$
\hat m_i(\Delta) = r^i |H_{\Delta_{i-1}} \sqcap \Delta| + \sum_{j=0}^{i-1} r^j |H_{\Delta_{j-1}} \sqcap \Delta_j|.  
$
By this definition $\hat m_0(\Delta) = |\hat H \sqcap \Delta|$, and we to bound $|\hat m_i(\Delta_i) - |\hat H \sqcap \Delta_i||$.  

\begin{lemma}
	\label{lem:walk}
	$|\hat{m}_{i}(\Delta) - |\hat H \sqcap \Delta|| = O(i \sqrt{|\hat{H}|\log \frac{2\ell}{\delta'}})$ with probability $1 - \delta'$.
\end{lemma}
\begin{proof}
By the triangle inequality we expand
\begin{align*}
|\hat m_i(\Delta) - \hat m_0(\Delta)| 
&\leq 
\sum_{j=0}^{i-1} |\hat m_{j+1}(\Delta) - \hat m_j(\Delta)|
\\ & =
\sum_{j=0}^{i-1} \left| \begin{array}{l}
\displaystyle{r^{j+1} |H_{\Delta_j} \sqcap \Delta| + \sum_{\ell=0}^j r^\ell |H_{\Delta_{\ell-1}} \sqcap \Delta_\ell|} 
\\ \phantom{12345} - 
\displaystyle{\left( r^{j} |H_{\Delta_{j-1}} \sqcap \Delta| + \sum_{\ell=0}^{j-1} r^{\ell-1} |H_{\Delta_{\ell-2}} \sqcap \Delta_{\ell-1}| \right)}
\end{array}\right|
\\ & =
\sum_{j=0}^{i-1} \left|
r^{j+1} |H_{\Delta_j} \sqcap \Delta| + r^j |H_{\Delta_{j-1}} \sqcap \Delta_j| 
- 
r^{j} |H_{\Delta_{j-1}} \sqcap \Delta|
\right|
\\ & =
\sum_{j=0}^{i-1} r^j \left|
r |H_{\Delta_j} \sqcap \Delta| - |(H_{\Delta_{j-1}} \cap \Delta_j) \sqcap \Delta |
\right|
\\ & \leq 
\sum_{j=0}^{i-1} r^j \cdot r \cdot C_d \sqrt{|H_{\Delta_j}| \log \frac{1}{\delta^{\dagger}}}.
\end{align*}
The last inequality follows since $H_{\Delta_j}$ is a random sample from $H_{\Delta_{j-1}} \cap \Delta_j$, and the $\sqcap \Delta$ restriction is a constant VC-dimension range~\cite{VC71}; the constant $C_d$ depends only on $d$, and $\delta^\dagger$ is the probability of failure for each term in the sum.  In particular we use
$
\left | \frac{|H_{\Delta_j} \sqcap \Delta|}{|H_{\Delta_j}|} - \frac{|(H_{\Delta_{j-1}} \cap \Delta_j) \sqcap \Delta|}{|H_{\Delta_{j-1}} \cap \Delta_j|} \right | 
\le 
C_d \sqrt{\frac{1}{|H_{\Delta_j}|} \log \frac{1}{\delta^\dagger}}
$
and multiply by $|H_{\Delta_{j-1}} \cap \Delta_j| = r |H_{\Delta_j}|$.

Applying Lemma \ref{lem:cellsize1} then $|H_{\Delta_j}| \le \frac{|\hat{H}|}{r^{2j}}$. Hence
\[
\sum_{j = 0}^{i - 1} r^j  \sqrt{|H_{\Delta_j}| \log \frac{1}{\delta^\dagger}} 
\le 
\sum_{j = 0}^{i - 1} r^j \sqrt{\frac{|\hat{H}|}{ r^{2j}} \log \frac{1}{\delta^\dagger}} 
\le
\sum_{j = 0}^{i - 1} \sqrt{|\hat{H}|\log \frac{1}{\delta^\dagger}}.
\]
And to ensure a failure probability of $1 - \delta'$ for the sequence of samples, set $\delta^\dagger = \delta'/2\ell$. 
\[
  |\hat{m}_{i}(\Delta) - |\hat H \sqcap \Delta|| \le i \cdot C_d \sqrt{|\hat{H}|\log \frac{2\ell}{\delta'}}. \qedhere
\]
\end{proof}

Now how large does $\hat{H}$ need to be to ensure a correct estimate of the cells at the leaves of the arrangement. This largely depends on the number of $\eps$-samples taken in total.

\begin{lemma}
	\label{lem:Hsize}
	We can ensure that $|\hat{m}_{i}(\Delta) - |\hat H \sqcap \Delta|| \le \eps$ for all $\Delta \in \mathbf{\Delta}_i$ with probability $\delta$ by setting $|\hat{H}| = O(\frac{i^3}{\eps^2} \log \frac{i}{\delta})$.  
\end{lemma}
\begin{proof}
	We can set the probability of $\delta' = \delta / (2|\mathbf{\Delta}_i|)$ in Lemma \ref{lem:walk} to ensure a failure probability of $\delta$ for each cell $\Delta \in \mathbf{\Delta}_i$ and therefore 
	$|\hat{m}_{i}(\Delta) - |\hat H \sqcap \Delta|| = O(i \sqrt{|\hat{H}|\log \frac{2 i |\mathbf{\Delta}_i|}{\delta}})$.
	So if $\eps |\hat{H}| = |\hat{m}_{i}(\Delta) - |\hat H \sqcap \Delta|| = O(i \sqrt{|\hat{H}|\log \frac{2 i |\mathbf{\Delta}_i|}{\delta}})$ then $|\hat H| = O(\frac{i^2}{ \eps^2}\log \frac{2i|\Delta_i|}{\delta})$.
	
	 From Lemma \ref{lem:cellsize} we know that $|\mathbf{\Delta}_i| = O(r^{d i})$, and so $|\hat{H}| = O(\frac{i^3}{\eps^2} \log \frac{i}{\delta})$.
\end{proof}

\subparagraph{Runtime.}
The time to compute an estimate for $\Delta \in \mathbf{\Delta}_i$ is linear in the number of lines in $H_{\Delta} \le \frac{|\hat{H}|}{r^{2i}}$.  
We can then devise the expected runtime for computing estimates for all cells in a level $i$, of which there are $\E[|\mathbf{\Delta_i}|] = O(r^{di})$, so the total for level $i$ is 
$\sum_{\Delta \in \mathbf{\Delta}_i} O(\frac{|\hat{H}|}{r^{2i}}) 
= 
O(\frac{|\mathbf{\Delta}_i||\hat{H}|}{r^{2i}}) 
= 
O(r^{i(d - 2)} |\hat{H}|)$. 
The total expected runtime over all levels is then 
$\sum_{i = 0}^{L} O(r^{i(d - 2)}|\hat{H}|)$. 
Furthermore, if $d = 2$ then this will be $O(\frac{L^4}{\eps^2} \log \frac{L}{\delta})$.   If $d > 2$ then each layer of the cutting will dominate the previous layers in runtime and so the total time will be bounded by the time to compute the last layer as  
$\sum_{i = 0}^{L} O(r^{i(d - 2)}|\hat{H}|) = O(r^{(L + 1)(d -2)}\frac{L^3}{\eps^2} \log \frac{L}{\delta})$.

Next we want to achieve that all cells $\Delta \in \mathbf{\Delta}_L$ are of size $|H_\Delta| < \log_r \frac{1}{\eps}$.  This can be achieved via $|H_\Delta| < |\hat H| / (r^{2L}) \le \log_r \frac{1}{\eps}$ by setting the maximum level $L > \frac{1}{2} \log_r \frac{|\hat H|}{\log_r \frac{1}{\eps}} = O(\log \frac{1}{\eps})$.  
Then a query can at first recursively descend into the structure for $O(L) = O(\log \frac{1}{\eps})$ steps, and upon reaching a leaf node, can enumerate the leaf node's sample which will take at most $O(\log\frac{1}{\eps})$ time again. 
With $L= O(\log \frac{1}{\eps})$, the total time to compute the cell division for $d= 2$ is $O(\frac{L^4}{\eps^2} \log \frac{L}{\delta}) = O(\frac{1}{\eps^2} \log^4 \frac{1}{\eps} \log \frac{\log \frac{1}{\eps}}{\delta})$. 
For $d > 2$, then $O(r^{(L + 1)(d -2)}\frac{L^3}{\eps^2} \log \frac{L}{\delta}) = O(\frac{1}{\eps^d} \log^3 \frac{1}{\eps} \log \frac{\log \frac{1}{\eps}}{\delta})$ since $r^{(\ell + 1)(d - 2)} = r r^{(d - 2)\log_r \frac{1}{\eps}} = r r^{\log_r \frac{1}{\eps^{d - 2}}}= \frac{r}{\eps^{d - 2}}$.

\begin{theorem}
\label{thm:ds-cost}
We can build an $\eps$-approximate halfspace range counting data structure of size $O(1/\eps^d)$ that for any halfspace $h \in \c{H}_d$ returns in $O(\log (1/\eps))$ time returns a count $\hat m(h)$ so that $|\hat m(h) - |h \cap X|| \leq \eps |X|$.  
The total expected construction time, with probability $1-\delta$, 
for $d = 2$ is $O(|X| + \frac{1}{\eps^2} \log^4 \frac{1}{\eps} \log \frac{\log \frac{1}{\eps}}{\delta})$, and 
for constant $d > 2$ is $O(|X| + \frac{1}{\eps^d} \log^3 \frac{1}{\eps} \log \frac{\log \frac{1}{\eps}}{\delta})$.
\end{theorem}

\subparagraph{Finding the Maximum Range.}
To query the structure we need a set of viable halfspaces that cover the space well enough to approximately hit the maximum region. We can use a random sample of points from the primal space of size $O(\frac{1}{\eps})$ to induce a set of $O(1/\eps^d)$ halfspaces $\hat{\c{H}}_d$, as was done in \cite{MP18b}, to get a constant probability that at least one halfspace is $O(\eps)$-close to the maximum region.  Then we repeat the procedure $\log \frac{1}{\delta}$ times and take the maximum found region to magnify the success probability to $1 - \delta$ (see \cite{MP18b} for details). 
For each query hyperplane $h$ we can query a structure constructed over $R$ and over $B$ and then compute the function value from this; we return the $h$ which maximizes $\Phi(h)$.  

The set of query hyperplanes from this tactic is of size $O(\frac{1}{\eps^d})$ for constant $d$ and at each level we can determine which cell the dual point of the halfspace falls into by testing a constant $O(r^d)$ number of constant sized cells.  At a leaf of the structure we check if the remaining, at most $\log \frac{1}{\eps}$, dual halfspaces are below the query dual point, to determine the total count. The query structure has $O(\log \frac{1}{\eps})$ levels and at each level a constant amount of work is done; it is repeated for each of $O(\frac{1}{\eps^d})$ halfspaces, and the entire endeavor is repeated $\log \frac{1}{\delta}$ times to reduce the probability of failure.   The full runtime is $O(\frac{1}{\eps^d} \log \frac{1}{\eps} \log \frac{1}{\delta})$,  plus the construction time of the data structure (from Theorem \ref{thm:ds-cost}) with dominates the cost.

\begin{theorem}
\label{thm:upperbnd}
We can solve \eMH with probability $1-\delta$, in expected time 
$O(|X| + \frac{1}{\eps^d} \log^4\frac{1}{\eps}\log \frac{\log \frac{1}{\eps}}{\delta})$ for $d = 2$ and 
$O(|X| + \frac{1}{\eps^d} \log^3 \frac{1}{\eps} \log \frac{\log \frac{1}{\eps}}{\delta})$ for constant $d > 2$.  
\end{theorem}

\section{Conditional Lower Bounds}
\label{sec:hs-LB}

Our upper bounds only restrict $\Phi_X(h) = \phi(\mu_R(h), \mu_B(h))$ to be a Lipschitz function.   Next we show that an algorithm that can operate on this entire class of functions for $d = 2$ has a conditional lower bound of $O(m + \frac{1}{\eps^{2 - o(1)}})$, depending on  \TSUM~\cite{GO95}.

However, the first bound requires $\phi$ is also concave (unlike $\phi_{||}$ or $\phi_K$ which are convex).  
So we consider a different convex function
$\phi(\mu_R,\mu_B) = \mu_R - \mu_B$ on weighted points; 
solving for it, and again after flipping all signs of weights, corresponds with $\phi_{||}$ (which can be used to approximate $\phi_K$).  We show \MH for this $\phi$ is lower bounded by $O(m^{3/2 - o(1)})$ conditional on 
\APSP~\cite{VWW18} requiring $\Omega(n^{3-o(1)})$ time.    



\subsection{Lower Bounds by 3SUM}
Gajentaan and Overmars identified a large class of problems in computational geometry  called \TSUM hard~\cite{GO95}.  \TSUM can be reduced to each of these problems, so an improvement in any one of them would imply an improvement in \TSUM.  While there are some loose lower bounds for this problem, 3SUM is conjectured to \textbf{not} be solvable in $O(m^{2 - o(1)})$ time \cite{KPP14,VW18}.
\TSUM reduces to the following problem in $O(m \log m)$ time~\cite{GO95}.  

\begin{itemize}
\item \textbf{Problem} \PC: \emph{From a given set of $m$ halfspaces determine if there is a point at depth $k$ ($k$ halfspaces lie above this point) where $k \le m/2$.}
\end{itemize}


Through standard point-line duality, we transform this to its equivalent dual problem.  
The first step is to use the well known point-line dual transform to convert points to lines and lines to points. Line $y = ax + b$ can be mapped to the point $(a, -b)$ and the point $(u, v)$ can be mapped to the line $y = ux - v$.  A line in the primal that lies above a point in the dual then its corresponding point will lie below the point's corresponding line. We will define this dual problem below.

\begin{itemize}
\item \textbf{Problem} \LC: \emph{From a given set of $m$ rays, oriented upwards or downwards, determine if there is a line cutting through $k$ rays where $k \le m/2$.}
\end{itemize}


Define the piecewise linear function (for $k \leq m/2$) on points $R,B \in \R^2$:
\[
\Phi(h) 
= 
\phi_{LC}(\mu_R,\mu_B) 
=
|R| - \left|\mu_R |R| - \mu_B |B| - k\right| 
=
|R| - \left||h \cap R| - |h \cap B| - k\right|. 
\]

\begin{lemma}
\LC is reducible to \MH in $\R^2$ with $\phi_{LC}$  in $O(m\log m)$ time.
\label{lem:hstolc}
\end{lemma}

\begin{proof}
Construct the lower envelope of all endpoints of the rays; this takes $O(m \log m)$ time.  
Each upwards oriented ray is replaced with a point at its end point, and placed in $R$.  
Each downward oriented ray is replaced with two points: its endpoint generates a point in $B$, and where it intersects the lower envelope generates a point in $R$.  
See Figure \ref{fig:pt-cover}.  

Lines now correspond to halfplanes below those lines.  
Upward ray intersections require lines above them.  
Downward rays require lines between the corresponding $B$ and $R$ points: if a line is above both, they cancel in $\phi_{LC}$; if it is below both, it includes neither; if a line is between them, it only identifies the $R$ point.  But the lines below both are below the lower envelope, and cannot be the optimal halfspace.  Thus scanning the generated point set $R \cup B$, if it identifies a halfspace $h$ where $\Phi(h) > |R|$, only then is \LC satisfied.  
\end{proof}

\begin{figure}
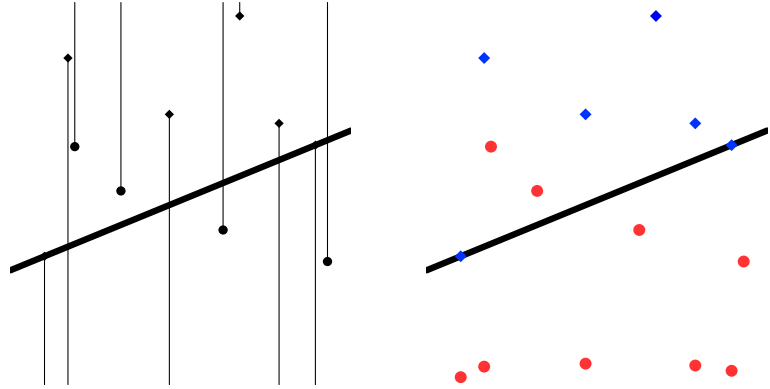

\centering
\includegraphics[width=.32\linewidth]{lower_bound1}
\hspace*{.3in}
\includegraphics[width=.32\linewidth]{lower_bound2}
\caption{On left converting the \PC problem into the dual makes it such that a line contains every ray that it intersects. On right the equivalent bichromatic discrepancy problem where red points have been placed on the lower envelope.}
\label{fig:pt-cover}
\end{figure}

We can also reduce the exact version to the approximate version. 
If we run an $\eps$-approximate \MH algorithm, and set $\eps = \frac{1}{2|R|}$ then the approximate range $h'$ found will be off by at most a count of $1/2$ from the optimal range, and hence must be the optimal solution.  

\begin{theorem}
	\label{thm:halfplane-LB}
In $\R^2$, \textsc{$\eps$-MaxHalfspace} for $\phi_{LC}$ takes $\Omega(m + 1/\eps^{2-o(1)})$ assuming the full input of size $m$ needs to be read, and \TSUM requires $\Omega(n^{2-o(1)})$ time.  
\end{theorem}	

Since $\phi_{LC}$ is concave and $1$-Lipschitz, this implies any algorithm that works for all concave $\phi$ or all $1$-Lipschitz $\phi$ must also take at least this long.

\subsection{Lower Bound by All Pairs Shortest Path}
We next provide a new construction that directly applies to the $\phi_{||}$ function on weighted points in $\R^2$ via a reduction from All Pairs Shortest Path (\APSP).  We first show the \MWTC problem reduces to \APSP via another problem \NT.  

\begin{itemize}
\item \textbf{Problem} \APSP: \emph{Given an edge-weighted undirected graph with $n$ vertices, find the shortest path between every pair of vertices.  }
\item \textbf{Problem} \NT: \emph{Given an edge-weighted undirected graph with $n$ vertices, determine if there is a triangle with negative total edge weight.}
\item \textbf{Problem} \MWKC: \emph{Given an edge-weighted undirected graph with $n$ vertices,  find the $K$-clique with the maximum total edge weight.}
\end{itemize}

While the \APSP is more well-known, Willams and Williams~\cite{VWW18} showed it is equivalent to the \NT problem, which will be more useful. Both have well-known $O(n^3)$ algorithms and are conjectured to not be solvable in less than $O(n^{3-o(1)})$ time, and improving that bound on one would improve on the other.  
Moreover, Backurs \etal~\cite{Backurs16} used \MWKC as a hard problem, believed to require $O(n^{K-o(1)})$ time, for which to reduce to several other problems involving rectangles.  They note that for $K=3$, \MWTC is a special case of \NT.  That is, for a guessed max-weight $\omega$, one can subtract $\omega/3$ from each edge weight, then multiply all weights by $-1$.  If there exists a negative triangle with the new weights, there exists a triangle with weight $\omega$ in the original.  One can resolve the max weight with logarithmically number of steps of binary search.  Hence, \NT (and hence also \APSP) reduces to \MWTC.

%
%
%
For \MWTC it is equivalent to assume a $3$-partite graph $G = (V, E)$~\cite{Backurs16}; that is, the vertices $V$ are the disjoint union of three independent sets $A = \{a_1, a_2, \ldots, a_n\}$, $B= \{b_1, b_2, \ldots, b_n\}$, and $C = \{c_1, c_2, \ldots, c_n \}$. 
Each independent set will have exactly $n$ vertices and each vertex, for instance $a_i$, will have an edge with every vertex in $B$ and $C$. Denote the edge between $a_i$ and $b_j$ as $e(a_i, b_j)$ and the weight as $w(e(a_i, b_j))$.  We reduce to a dual of halfspace scanning in $\R^2$:

\begin{itemize}
\item \textbf{Problem} \MWP: \emph{Given $m$ weighted lines find a point which maximizes the sum of weights of all lines passing below that point (or intersecting that point).  }
\end{itemize}

Our reduction will rely on a planar geometric realization of any such graph $G$ with $m = O(n^2)$ lines, where the lines correspond to edges, and triple configurations of lines correspond with cliques.  
Given such an instance, if we can solve the \MH algorithm (in the dual as \MWP) in better than $O(m^{\frac{3}{2} - o(1)})$ time we can recover the solution to the \MWTC problem in better than $O(n^{3-o(1)})$ time.

\subparagraph{The double weighted line gadget.}
Our full construction will use a special gadget which will ensure the max weighted point will correspond with a vertex at a triple intersection of a planar arrangement of lines, with each line corresponding with an edge in $G$.  For this to hold we instantiate each edge as a \emph{double line} $d(e)$.  This consists of two parallel lines $\ell_u(e)$ and $\ell_l(e)$, separated vertically by a gap of a small positive value $\alpha \ll 1$, with $\ell_u(e)$ above $\ell_l(e)$.    In what follows we will refer to a double line as a single line; to be precise this will refer to a line at the midpoint between $\ell_l$ and $\ell_u$.  
Let $\bar{w} = 2 \max_{e \in E} |w(e)| + 1$ be a large enough value so for any $e \in E$ that $\bar w + w(e)$ is strictly positive, and any $|w(e)| < \bar w /2$.  
Then we denote the weight of the double line as $w(d(e)) = w(e) + \bar w$.  
This is transferred to the lines as $w(\ell_u(e)) = -w(d(e))$ and $w(\ell_l(e)) =  w(d(e))$.  

Now given a query point $x \in \R^2$ we say it is \emph{on} $d(e)$ if it lies between lines $\ell_l(e)$ and $\ell_{u}(e)$, and \emph{not on} otherwise.  This will allow us to control the effect of $e$ on query $x$ in a precise way.  

\begin{lemma}
Any point $x$ that lies on $d(e)$ will have weight contributed to it by $e$ of exactly $w(e) + \bar w$; otherwise that contribution will be $0$.  
\end{lemma}

\subparagraph{Reduction to triple intersections.}
Our construction will place double lines for each edge in the graph.  
The edges from $A$ to $B$ will be \emph{blue} lines; from $A$ to $C$ will be \emph{red} lines, and from $B$ to $C$ will be \emph{black} lines.  
All lines of the same color will be parallel; this will ensure that any query point $x \in \R^2$ can only be on one double line of each color.  
For easy of exposition, in our construction description (and illustration; see Figure \ref{fig:lower-construction}) the black lines will be vertical, which makes ambiguous the ``above'' relation; so the final step will be to rotate the entire construction clockwise by a small angle (less than $\pi/4$ radians).

The construction will now lay out the double lines so that every clique $\{a_i, b_j, c_k\}$ will be realized as a triple intersection of double lines $d(e(a_i,b_j))$, $d(e(a_i,c_k))$, and $d(e(b_j,c_k))$ and so there are no other types of triple intersections.  Such a triple intersection will have weight precisely $w(e(a_i,b_j)) + w(e(a_i,c_k)) + w(e(b_j,c_k)) + 3 \bar w > (3/2) \bar w $, and any other point (e.g., a double intersection) must have weight strictly less than $(3/2) \bar w$.    

\begin{lemma}
The maximum weight point must occur at a triple intersection of three double lines of different colors, and thus must correspond the max weight $3$-clique.  
\end{lemma}

\begin{figure}
\centering
\includegraphics[width=.55\linewidth]{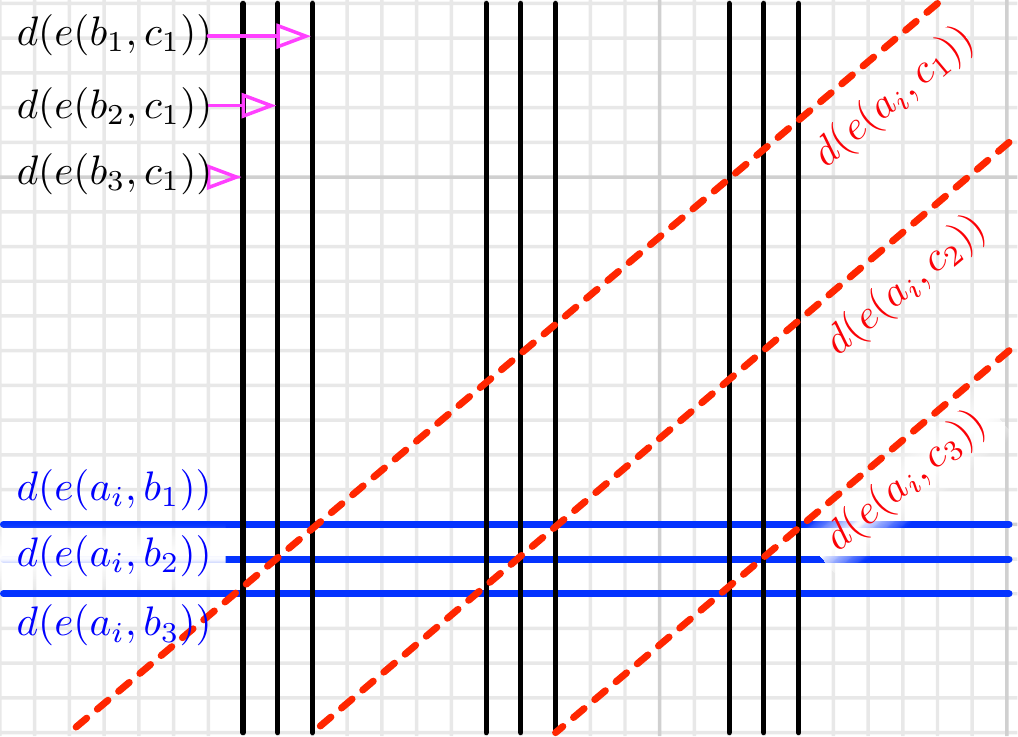}
\includegraphics[width=.44\linewidth]{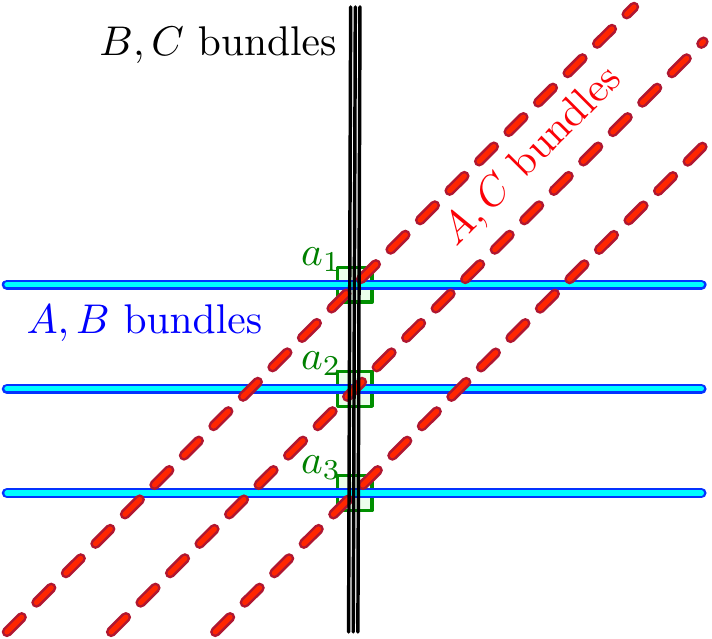}
\caption{\MWP construction with $n=3$.  Left: intersection of $a_i$ blue and red bundles with all black double lines.  Right: relation of the bundles, with 3 $a_i$ intersections marked.}  
\label{fig:lower-construction}
\end{figure}

\subparagraph{The full construction.}
The blue lines (using $A$ to $B$ edges) will all be horizontal.  Edge $e(a_i, b_j)$ will have $y$-coordinate $y_{i,j} = -j - i (5n^2)$ so that $y_{i,j+1} = y_{i,j}-1$ and so $y_{i+1,j} = y_{i,j} - 3n^2$.  Set $y_{1,1} = -5n^2-1$.  This bundles the blue lines associated with the same $a_i$ point.  Figure \ref{fig:lower-construction} shows a single bundle (left) and structure of all bundles (right).  

The red lines (using $A$ to $C$ edges) will be at a $45$ degree angle (a slope of $1$), similarly clustered by their $c_k$ values.  Double line $d(e(a_i,c_k))$ will have equation $\textbf{y} = \textbf{x} + o_{i,k}$.  
We define the offsets $o_{i,k} = -k(3n) - i(5n^2)$ so that $o_{1,1} = -3n-5n^2$; $o_{i,k+1} = o_{i,k} - 3n$ and $o_{i+1,k} = o_{i,k} - 5n^2$.   
And in particular, double lines $d(e(a_i,b_j))$ (at horizontal $y_{i,j}$) and $d(e(a_i,c_k))$ (at offset $o_{i,k})$ will intersect at $x$-value $x_{j,k} = y_{i,j} - o_{i,k}$.  

The black lines (using $B$ and $C$) will be vertical.  Edge $d(e(b_j,c_k))$ will have $x$ coordinate $x_{j,k}$, defined $x_{j,k} = y_{i,j} - o_{i,k} = -j + (3n)k$.  
Also, these $x_{j,k}$ values will be distinct for different values of $j,k$, but independent of the choice of $a_i$.  
Moreover, these are the same $x$-coordinates where the corresponding red and blue lines intersect.  Thus each black line $d(e(b_j,c_k))$ intersects the intersection of blue line $d(e(a_i,b_j))$ and $d(e(a_i,c_k))$ for each $a_i$.  

Finally, note that all black lines are in the $x$ range $[-3n^2, 0]$ we can argue that they do not cause any other triple intersections.  Because the each red bundle has offsets separated by more than $3n^2$ then a red bundle associated with $a_i$ cannot intersect a blue bundle associated $a_{i'}$ for $i \neq i'$ since the red lines have linear slope and the blue bundles are also separated by more than $3n^2$.  Thus the intended triple intersections (of $d(e(a_i,b_j))$, $d(e(a_i,c_k))$, and $d(e(b_j,c_k))$) are the only ones in this construction.  

In total there are $n^2$ blue, $n^2$ blue, and $n^2$ black double lines, thus $m = O(n^2)$.  
Hence, \MWTC on $n$ vertices reduces to \MH in time $O(n^2)$. 
Then reversing the dual mapping, we consider each dual line as two points, one in $R$ and one in $B$, and this corresponds with the \MH problem in $\R^2$ with $\phi_{||}$.  
Then since \APSP reduces to \MWTC, we obtain the following theorem.  

\begin{theorem}
In $\R^2$, \MH for $\phi_{||}$ on $m$ points requires $\Omega(m^{3/2 - o(1)})$ time assuming that \APSP on $n$ vertices requires $\Omega(n^{3 - o(1)})$ time.  
\label{parttohalfplane}
\end{theorem}

\section{Conclusions and Discussion}
We have mostly closed the planar \eMH problem with an $\tilde O(m + 1/\eps^{d})$ algorithm in $\R^d$ and conditional (to \TSUM) $\Omega(m + 1/\eps^{2-o(1)})$ lower bound in $\R^2$.  However, the lower bound uses a piecewise-linear function $\phi_{LC}$, and while all known algorithmic improvements that depend on $\phi$ take advantage of this linear structure (e.g., for rectangles~\cite{APV06,MP18b}), the function $\phi_{LC}$, perhaps strangely, is concave.  
Also surprisingly we can prove another conditional lower bound for \MH using a convex (in fact, again linear) function $\phi$, but this one is smaller at $\Omega(m^{3/2 - o(1)})$ in $\R^2$, and because some point may have very small weight in the construction, does not directly apply to \eMH, which allows $\eps$ additive error.  Moreover, this reduces to \APSP, which does not appear to be in the same class as \TSUM~\cite{VW18}.

We leave several curious aspects to future work.  Is there a real fine-grained complexity difference between the $\phi$ variants in the problems?  That is, can we improve upper bounds for convex $\phi$, or improve conditional lower bounds in this case?  And can we condition the results on \TSUM in the convex $\phi$ case?  The convex $\phi$ lower bound construction relies on weighted points, can we obtain improved algorithmic runtime by only allowing $\{-1,+1\}$ weights?  
Moreover, are the polynomial terms $(1/\eps)^d$ in the algorithmic runtime correct in constant dimensions larger than $d=2$?  And can these results help resolve polynomial terms in the high-dimensional robust statistics settings?


\bibliography{bib-disc}
\normalsize

\end{document}